\documentclass[letterpaper, 10 pt, journal, twoside]{IEEEtran}
\usepackage{graphics} 
\usepackage{epsfig} 
\usepackage{mathptmx} 
\usepackage{times} 
\usepackage{amsmath} 
\usepackage{amssymb}  
\usepackage{mathtools}

\usepackage{bm}
\usepackage{accents}

\usepackage[caption=false,font=footnotesize]{subfig}
\usepackage{float}

\usepackage{tabularx}
\usepackage{booktabs} 

\usepackage{color}

\usepackage{dsfont}

\usepackage{amsthm}

\newtheorem{theorem}{Theorem}

\newtheorem{lemma}{Lemma}
\newtheorem{definition}{Definition}
\newtheorem{assumption}{Assumption}
\newtheorem{remark}{Remark}
\newtheorem*{problem}{Problem Statement}

\usepackage[linesnumbered,ruled,vlined]{algorithm2e}
\SetKwInput{KwInput}{Input}
\SetKwInput{KwOutput}{Output}

\usepackage{xcolor}

\usepackage{graphicx}

\SetAlCapNameFnt{\small}
\SetAlCapFnt{\small}

\usepackage[hidelinks, bookmarks=true,pagebackref=false]{hyperref}
\usepackage[capitalise,nameinlink]{cleveref}

\usepackage{mdframed}
\newmdtheoremenv{theo}{Theorem}

\usepackage{flushend}

\usepackage[normalem]{ulem}

\begin{document}
%
\title{Set-based State Estimation with Probabilistic Consistency Guarantee under Epistemic Uncertainty}
%
%
%

\author{Shen Li*$^{1}$, Theodoros Stouraitis*$^{1,2}$, Michael Gienger$^{3}$, Sethu Vijayakumar$^{2,4}$, and Julie A. Shah$^{1}$
\thanks{Manuscript received: October, 6, 2021; Revised January, 14, 2022; Accepted February, 7, 2022. This paper was recommended for publication by Editor Jens A. Kober upon evaluation of the reviewers' comments.
This research was supported by
the U.S.A. Office of Naval Research (Sponsor Award ID N00014-18-1-2815),
the Alan Turing Institute,
and the Honda Research Institute Europe. (Corresponding author: Shen Li)} 
\thanks{*Both authors contributed equally.}
\thanks{$^{1}$Authors are with CSAIL, Massachusetts Institute of Technology, U.S.A. {\tt\small \{shenli, theostou\}@mit.edu}}%
\thanks{$^{2}$Authors are with School of Informatics, University of Edinburgh, U.K..}
\thanks{$^{3}$Author is with Honda Research Institute Europe, Germany.}
\thanks{$^{4}$Author is with Alan Turing Institute, U.K..} 
\thanks{Digital Object Identifier (DOI): see top of this page.}
}
%
%

\markboth{IEEE Robotics and Automation Letters. Preprint Version. Accepted February, 2022}
{Li \MakeLowercase{\textit{et al.}}: Set-based State Estimation with Probabilistic Consistency Guarantee} 

%



\maketitle

\begin{abstract}
Consistent state estimation is challenging, especially under the epistemic uncertainties arising from learned (nonlinear) dynamic and observation models.
In this work, we propose a set-based estimation algorithm, named Gaussian Process-Zonotopic Kalman Filter (GP-ZKF), that produces zonotopic state estimates while respecting both the epistemic uncertainties in the learned models and aleatoric uncertainties.
Our method guarantees probabilistic consistency, in the sense that the true states are bounded by sets (zonotopes) across all time steps, with high probability.
We formally relate GP-ZKF with the corresponding stochastic approach, GP-EKF, in the case of learned (nonlinear) models.
In particular, when linearization errors and aleatoric uncertainties are omitted and epistemic uncertainties are simplified, GP-ZKF reduces to GP-EKF.
We empirically demonstrate our method's efficacy in both a simulated pendulum domain and a real-world robot-assisted dressing domain, where GP-ZKF produced more consistent and less conservative set-based estimates than all baseline stochastic methods.
\end{abstract}


%
\IEEEpeerreviewmaketitle


\section{Introduction}\label{sec:introduction}
 
\IEEEPARstart{S}{tate} estimation is critical for robot decision making, especially during human-robot interactive tasks, where physical states can be partially observed due to occlusions~\cite{zhang2019probabilistic} and latent mental states can influence the interaction~\cite{xu2015optimo,chen2020trust}.
For a robot to safely interact with a human, such states are typically estimated based on models~\cite{choudhury2019utility}.
Due to the complexity in model specification, human behavior and observation models are usually learned from data~\cite{unhelkar2020decision}.
These learned models typically include errors, and any process has stochastic noise; the former is described as ``epistemic uncertainty,'' and the latter as ``aleatoric uncertainty,'' both of which present critical challenges for estimation algorithms.
For example, in a robot-assisted dressing scenario (see Fig.~\ref{fig:teaser}), the epistemic uncertainty could blind the robot, resulting in overconfidence in an erroneous human state and rendering the robot's behavior ``\textit{aggressive}'' and unsafe.
To address this challenge, we developed a set-based estimation algorithm that is able to \textit{conservatively} respect these uncertainties.

We focus on the problem of consistency in state estimation from the view of epistemic uncertainty introduced by learned models.
In prior literature, consistency has been analyzed within two paradigms: stochastic paradigm and set-based paradigm.
For stochastic methods, such as the Extended Kalman Filter (EKF),
a consistent estimate is defined as an unbiased point estimate together with a covariance matching the actual estimation error~\cite{castellanos2004limits}.
In the field of SLAM, the inconsistency issue of EKF-based approaches, such as GP-EKF~\cite{ko2009gp}, has been broadly investigated from the views of linearization errors~\cite{julier2001counter,huang2008analysis} and state unobservability~\cite{huang2007convergence,barrau2015ekf}.
%
On the other hand, set-based methods construct sets, as state estimates, where consistency can be interpreted as the sets containing the true states~\cite{alamo2005guaranteed}---\textit{i.e.}, bounded estimation error.
For known models, where \textit{epistemic} uncertainty can be omitted, set-based methods can produce estimates with guaranteed consistency under aleatoric uncertainty~\cite{alamo2005guaranteed, rego2020guaranteed}.
%
In contrast, we consider settings in which nonlinear models are learned; in such circumstances, uncertainty arises from both the learning errors (\textit{epistemic}) and noises (\textit{aleatoric}).

\begin{figure}
  \centering
  \includegraphics[height=53mm, width=0.99\columnwidth]{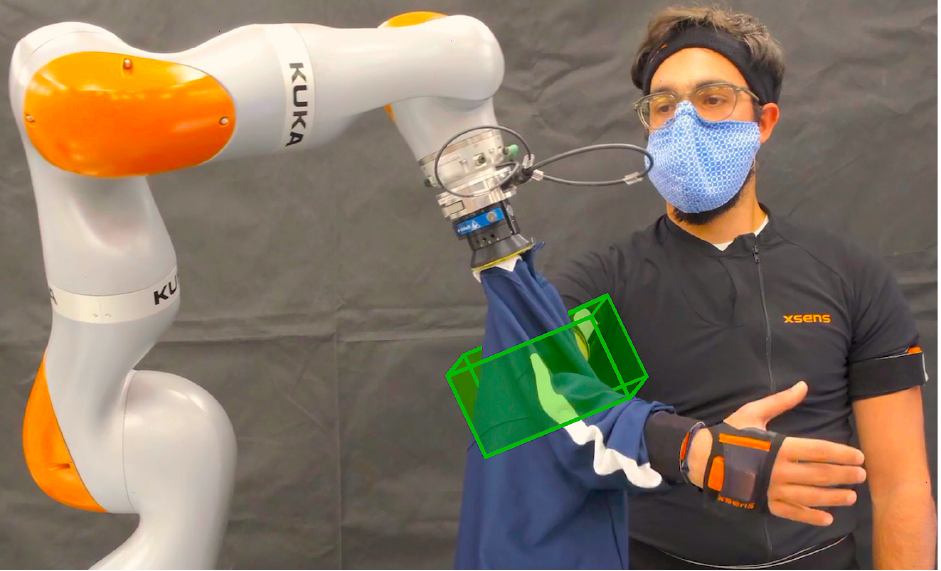}
  \caption{
  In a robot-assisted dressing scenario, we deployed our set-based estimator---GP-ZKF---to estimate the visually occluded human elbow position~\cite{zhang2019probabilistic}.
  With human dynamic and observation models learned via Gaussian Processes, GP-ZKF constructs zonotopic state estimates (illustrated with the green box) based on the force measurements at the robot end effector.
  By handling the epistemic uncertainties in the learned models, GP-ZKF guarantees probabilistic consistency---\emph{i.e.}, the true human elbow positions are bounded by the zonotopes across all time steps, with high probability.}
  \vspace{-2mm}
  \label{fig:teaser}
  \vspace{-4mm}
\end{figure}


We developed Gaussian Process-Zonotopic Kalman Filter (GP-ZKF), a set-based estimation algorithm with a probabilistic consistency guarantee under learned (dynamic and observation) models.
GP-ZKF learns both nonlinear models via Gaussian Processes (GP) and leverages GP's confidence intervals~\cite{koller2019arXivlearning} to calibrate the epistemic uncertainties in both models.
This extends prior work~\cite{combastel2016extended} that assumed bounded epistemic uncertainties in the linear parameter-varying enclosures of the nonlinear models.
Similar to set-based estimators~\cite{alamo2005guaranteed, rego2020guaranteed}, but specifically for scenarios with learned models, our approach recursively produces set-based estimates that are represented as zonotopes (a special type of polytope).
These zonotopes are designed to respect both epistemic and aleatoric uncertainties and guaranteed to contain the true states across all times steps, with high probability, rendering GP-ZKF consistent when both nonlinear models are learned.

We formally relate our \textit{set-based} approach---GP-ZKF---with the corresponding \textit{stochastic} approach---GP-EKF~\cite{ko2009gp}---and prove that
GP-ZKF reduces to GP-EKF if GP-ZKF omits linearization errors and aleatoric, and simplifies epistemic uncertainties.
This theoretical connection under nonlinear and learned models extends prior work~\cite{combastel2015zonotopes} in relating set-based with stochastic paradigms under linear and known models.

In this paper, we make the following contributions:
\begin{itemize}
  \item
  We propose GP-ZKF, a set-based state estimator with a probabilistic consistency guarantee under epistemic and aleatoric uncertainties, in the case where both dynamic and observation models are nonlinear and learned.
  \item
  We formally relate GP-ZKF with its stochastic counterpart, GP-EKF~\cite{ko2009gp}, under nonlinear and learned models.
\end{itemize}

We evaluated GP-ZKF in a simulated pendulum domain and a real-world robot-assisted dressing domain; our results show that GP-ZKF provides not only more consistent, but less conservative set-based estimates than the stochastic baselines (GP-EKF, GP-UKF, and GP-PF~\cite{ko2009gp}).
To the best of our knowledge, ours is the first work with a probabilistic consistency guarantee for state estimation with learned models.

In Sec.~\ref{sec:background}, we provide the background, in Sec.~\ref{sec:system_formulation} we present the system formulation, and in Sec.~\ref{sec:problem_def} we formally define probabilistic consistency---our focal point for this work. We present our method in Sec.~\ref{sec:method}, prove its two main theorems in Sec.~\ref{sec:theory}, and empirically evaluate it in Sec.~\ref{sec:experiment_and_result}.

\section{Background}\label{sec:background}

\begin{figure*}[t]
    \centering
    \vspace{6mm}
    \def\svgwidth{0.975\linewidth}
    \input{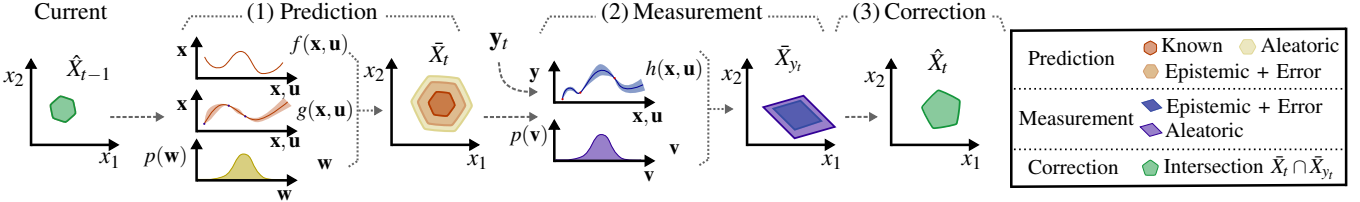}
    \vspace{-1.5mm}
    \caption{A flowchart illustrating the three phases of the set-based estimator, GP-ZKF, at time $t=1,\dots,T$: 
    (1) Prediction: given the previous zonotopic estimate, $\hat{X}_{t-1}$, and the control, $\mathbf{u}_{t-1}$ (omitted in the figure), GP-ZKF produces a dynamics-consistent zonotope, $\bar{X}_{t}$.
    (2) Measurement: given $\bar{X}_{t}$, the new measurement, $\mathbf{y}_{t}$, and the control, $\mathbf{u}_t$ (omitted in the figure), GP-ZKF produces a measurement-consistent polytope, $\bar{X}_{y_t}$.
    (3) Correction: GP-ZKF computes the new zonotopic estimate, $\hat{X}_t$, via $\bar{X}_{t} \cap \bar{X}_{y_t}$.
    In the prediction phase, the dynamics contains the known function, $f(\cdot)$; the learned function, $g(\cdot)$; and the process noise, $\mathbf{w}$ (Eq.~\ref{eq:system:dynamics}).
    In the measurement phase, the observation function contains the learned function, $h(\cdot)$, and the measurement noise, $\mathbf{v}$ (Eq.~\ref{eq:system:measurement}).
    }
    \label{fig:algo_conc}
    \vspace{-2mm}
\end{figure*}

Here, we introduce our nomenclature and briefly provide the background on intervals, zonotopes, Gaussian Process (GP), and a high-probability bound for Gaussian noise.

\textit{Nomenclature}:
Uppercase symbols denote sets, bold upper case symbols denote matrices, lowercase symbols denote scalars, and bold lowercase symbols denote vectors.
Subscripts denote the time and dimension:
\textit{e.g.}, $\mathbf{x}_{t,j}$, for example, is the $j$th dimension of a vector, $\mathbf{x}$, at time $t$. 
Superscripts denote the function a variable is related to---\textit{e.g.}, $\mu^g$ is the mean of the function $g(\cdot)$.
The Minkowski sum, $\bigoplus$, between two sets, $X$ and $Y$, is defined as $X \bigoplus Y \coloneqq \{\mathbf{x}+\mathbf{y} \colon \mathbf{x} \in X, \mathbf{y} \in Y\}$.

\subsection{Intervals and Zonotopes}\label{background:zonotope}
An interval $[a,b]$ is defined as $\{x\colon a \leq x \leq b\}$.
A box $\subset \mathbb{R}^n$ is an interval vector, $([a_1,b_1],\dots,[a_n,b_n])^T$.
A zero-centered box with radius $\mathbf{r} \in \mathbb{R}^n$ is defined as $[\mathbf{0} \pm \mathbf{r}] \subset \mathbb{R}^n$.

Zonotopes are convex polytopes that are centrally symmetric~\cite{scott2016constrained}.
A set $Z \in \mathbb{R}^n$ is a zonotope if there exists $(\mathbf{G}_Z, \mathbf{c}_Z) \in \mathbb{R}^{n \times n_\xi} \times \mathbb{R}^{n}$, such that $Z=\{\mathbf{c}_Z+\mathbf{G}_Z \mathbf{\xi} \colon \mathbf{\xi} \in \mathbb{R}^{n_\xi}, \|\mathbf{\xi}\|_\infty\leq 1\}$.
Here, $\mathbf{c}_Z$ is the center, each column of $\mathbf{G}_Z$ is a generator, and $\mathbf{\xi}$ contains all generator variables.
In this work, we let $(Z)_{c}$ refer to the center, $\mathbf{c}_Z$, and let $(Z)_{G}$ refer to the generators, $\mathbf{G}_Z$.
We compactly denote $Z = \{\mathbf{G}_Z,\mathbf{c}_Z\}$.
Zonotopes are closed under affine transformations and Minkowski sums, both of which can be computed exactly.
Formally, $\mathbf{b} \bigoplus \mathbf{A}Z=\{\mathbf{A}\mathbf{G}_Z, \mathbf{b}+\mathbf{A}\mathbf{c}_Z\}$ and $Z_1 \bigoplus Z_2=\{[\mathbf{G}_{Z_1} \:\: \mathbf{G}_{Z_2}], \mathbf{c}_{Z_1}+\mathbf{c}_{Z_2}\}$.
For details about intervals and zonotopes, please refer to~\cite{alamo2005guaranteed,rego2020guaranteed,  althoff2010reachability}.

\subsection{Gaussian Process and Confidence Intervals}\label{background:gp}

In this work, unknown functions are learned via GPs.
First, for any unknown function $\psi(\mathbf{x}) \colon \mathbb{R}^{n_x} \rightarrow \mathbb{R}^{n_y}$, we assume that along all dimensions, $j \in J$ with $J = \{1,\dots,n_y\}$, the function outputs, $\psi_j(\cdot)$, are independent of each other.
Then, we equivalently reformulate the multi-output function, $\psi(\cdot)$, using a single-output \textit{surrogate} function, $\psi'(\cdot)$~\cite{koller2019arXivlearning}.
We define $\psi'(\cdot) \colon \mathbb{R}^{n_x} \times J \rightarrow \mathbb{R}$, where for each dimension $j$, we have $\psi'(\mathbf{x}, j) = \psi_j(\mathbf{x})$;
this allows us to apply the standard definition of GP with a scalar output and formulate confidence intervals.

We use a GP, denoted by $\mathit{GP}(m^\psi,k^\psi)$, to learn $\psi'(\cdot)$, where the prior mean function, $\mu^\psi(\cdot)$, is set to $0$.
We assume that the given training data points are corrupted by \textit{i.i.d.} Gaussian noise $\mathbf{\upsilon} \in \mathbb{R}^{n_y}$.
Formally, for each dimension $j$, the noise $\mathbf{\upsilon}_j \sim \mathcal{N}(0,\lambda_\upsilon^2)$, where $\lambda_\upsilon \in \mathbb{R}$.
By conditioning the GP on $n$ data points, we obtain that, for each $j$, a posterior mean function, $\mu^\psi_{n,j}(\cdot)$, and a posterior variance function, $(\sigma^\psi_{n,j})^2(\cdot)$.

This work assumes that the true function, $\psi'(\cdot)$, belongs to the reproducing kernel Hilbert space (RKHS) associated with the kernel, $k^\psi$.
The smoothness of $\psi'(\cdot)$ can be measured via its RKHS norm, denoted by $\|\psi'\|_{k^\psi}$~\cite{koller2019arXivlearning};
then, the confidence intervals for the function outputs of $\psi'(\cdot)$---or, equivalently, of $\psi(\cdot)$---can be formulated as follows:
\begin{lemma}[Lemma 1 in \cite{bogunovic2020corruption}]\label{lemma:well_calibrated}
    Let $\delta^\psi \in (0,1)$.
    Fix $\psi$ in RKHS with $\|\psi'\|_{k^\psi} \leq B^\psi$.
    Assume that each dimension of the data of $\psi'(\cdot)$ is corrupted by \textit{i.i.d.} Gaussian noise $\mathbf{\upsilon}_j \sim \mathcal{N}(0,\lambda_\upsilon^2)$.
    Let $\beta_n^\psi=B^\psi + \sqrt{2 (\gamma_{n \cdot n_y}+\log(1/\delta^\psi))}$.
    Then, with probability at least $(1-\delta^\psi)$,
    jointly for all $n \in \mathbb{N}$,  $j = 1,\dots,n_y$, $\mathbf{x} \in \mathbb{R}^{n_x}$, we have that
    $
    |\mu^\psi_{n, j}(\mathbf{x}) - \psi_j(\mathbf{x})| \leq \beta_n^\psi \cdot \sigma^\psi_{n, j}(\mathbf{x})\nonumber
    $.
\end{lemma}
This lemma states that with high probability, jointly for each dimension $j$, the function output of $\psi_j(\cdot)$ is bounded by a confidence interval centered at the posterior mean prediction $\mu^\psi_{n,j}$.
The scaling, $\beta_n^\psi$, depends upon $\gamma_{n \cdot n_y}$, the information capacity with $n \cdot n_y$ data points for $\psi'(\cdot)$,
which can be bounded when the domain of $\psi'(\cdot)$ is compact~\cite{koller2019arXivlearning,srinivas2012information}.

\subsection{Gaussian Noise Bound}\label{background:gaussian_noise}
We consider \emph{i.i.d.} Gaussian noise, denoted by $\mathbf{\upsilon}_t \in \mathbb{R}^{n_y}$, for $t=1,\dots,T$, where $T \in \mathbb{N}$ denotes a finite time horizon.
With high probability, jointly throughout all time steps, noise can be bounded by a box, $\subset \mathbb{R}^{n_y}$.
Formally:
\begin{lemma}\label{lemma:bound_gaussian_noise}
  Let vectors $\mathbf{\upsilon}_1,\dots,\mathbf{\upsilon}_{T} \in \mathbb{R}^{n_y}$, such that for each time $t = 1, \dots, T$ and dimension $j = 1,\dots,n_y$, the noise $\upsilon_{t,j} \sim \mathcal{N}(0,\lambda_\upsilon^2)$, where $\lambda_\upsilon \in \mathbb{R}$.
  Then, with probability at least $(1-\delta^\upsilon)$, where $\delta^\upsilon \in (0,1)$, jointly for all $t = 1, \dots, T$, we have that
  $\mathbf{\upsilon}_{t} \in \left[0 \pm \sqrt{2} \lambda_\upsilon \sqrt{\ln (T\cdot n_y/\delta^\upsilon)} \right]^{n_y} \subset \mathbb{R}^{n_y}$.
\end{lemma}

\begin{proof}
The proof is similar to those for Lemma 5.1 in \cite{srinivas2012information} and Lemma 4 in \cite{berkenkamp2019safe}. 
For each time $t=1,\dots,T$ and dimension $j=1,\dots,n_y$, we bound $\mathbf{\upsilon}_{t,j}$ by applying the Gaussian error function with a probability budget, $\delta^\upsilon/(T\cdot n_y)$;
we then obtain the result via a union bound over all $t$ and $j$.
\end{proof}

\section{System Formulation}\label{sec:system_formulation}
We consider a discrete-time dynamical system with finite-horizon $T \in \mathbb{N}$, nonlinear dynamics and observation functions, and additive noises.
Formally, for $t = 1, \dots, T$, the system can be described as follows:
\begin{align}
  &\mathbf{x}_t = d(\mathbf{x}_{t-1}, \mathbf{u}_{t-1}, \mathbf{w}_{t-1})\nonumber\\
  &\quad= f(\mathbf{x}_{t-1}, \mathbf{u}_{t-1}) + g(\mathbf{x}_{t-1}, \mathbf{u}_{t-1}) + \mathbf{w}_{t-1}\label{eq:system:dynamics}\\
  &\mathbf{y}_t = o(\mathbf{x}_{t}, \mathbf{u}_{t}, \mathbf{v}_{t}) = h(\mathbf{x}_{t}, \mathbf{u}_{t}) + \mathbf{v}_{t}\label{eq:system:measurement}
\end{align}
Here, for all $t = 0, \dots, T$, $\mathbf{x}_t \in \mathbb{R}^{n_x}$, $\mathbf{u}_t \in \mathcal{U} \subset \mathbb{R}^{n_u}$, $\mathbf{y}_t \in \mathbb{R}^{n_y}$, $\mathbf{w}_t \in \mathbb{R}^{n_x}$, and $\mathbf{v}_t \in \mathbb{R}^{n_y}$ denote the system state, control, measurement,
process noise, and measurement noise, respectively.
We assume that the domain $\mathcal{U}$ is compact, and each element of $\mathbf{w}_t$ and $\mathbf{v}_t$ are \emph{i.i.d.} Gaussian.
Formally, for each dimension $j$, 
we have $\mathbf{w}_{t,j} \sim \mathcal{N}(0, \lambda_w^2)$ and $\mathbf{v}_{t,j} \sim \mathcal{N}(0, \lambda_v^2)$, where $\lambda_w,\lambda_v \in \mathbb{R}$.
Here, $\mathbf{w}_{t+1}$ and $\mathbf{v}_t$ correspond to $\mathbf{\upsilon}_t$ in Sec.~\ref{background:gp}.
The function $f(\cdot)$ denotes the prior known dynamic model, which can be a parametric physics model.
The functions $g(\cdot)$ and $h(\cdot)$
denote the unknown functions to be learned via GP, and correspond to $\psi(\cdot)$ in Sec.~\ref{background:gp}.

\begin{remark}\label{remark:noise_bounds}
By applying Lemma~\ref{lemma:bound_gaussian_noise}, we construct boxes, denoted by $W \subset \mathbb{R}^{n_x}$ and $V \subset \mathbb{R}^{n_y}$, to bound noises $\mathbf{w}$ and $\mathbf{v}$, respectively.
Formally, given $\delta^w \in (0,1)$, with probability at least $(1-\delta^w)$, jointly for all $t=1,\dots,T$, we have that $\mathbf{w}_{t-1} \in W$.
Given $\delta^v \in (0,1)$, with probability at least $(1-\delta^v)$, jointly for all $t=1,\dots,T$, we have that $\mathbf{v}_{t} \in V$.
\end{remark}


\section{Problem Definition}\label{sec:problem_def}
Our goal is to develop a set-based state estimator that can produce \textit{consistent} estimates.
Let $\hat{X}_{t} \subset \mathbb{R}^{n_x}$ denote a set-based state estimate produced by our algorithm at time $t$.
By assuming that the controls are given, the estimation process, at time $t=1,\dots,T$ can be represented as a recursive function, $\hat{X}_{t} = E(\hat{X}_{t-1}, \mathbf{u}_{t-1}, \mathbf{u}_t, \mathbf{y}_t)$.

When the dynamic and observation models, $d(\cdot)$ and $o(\cdot)$ (respectively), are known, prior arts in set-based state estimation~\cite{alamo2005guaranteed, rego2020guaranteed,combastel2005state} have achieved \textit{strict} consistency guarantees.
In contrast, we focus on a scenario where both $d(\cdot)$ and $o(\cdot)$ are learned with a limited amount of data; hence, we relax the strict consistency and focus on \textit{probabilistic} consistency, or $\delta$-consistency, defined as follows:
\begin{definition}[$\delta$-consistency]\label{def:delta_consistency}
  Given $\delta \in (0, 1)$;
  an initial set-based estimate, $\hat{X}_0 \subset \mathbb{R}^{n_x}$ such that $\mathbf{x}_0 \in \hat{X}_0$;
  a sequence of controls, $\{\mathbf{u}_t\}_{t=0}^{T} \subset \mathcal{U}$;
  and a sequence of measurements, $\{\mathbf{y}_t\}_{t=1}^{T} \subset \mathbb{R}^{n_y}$;
  Then, a state estimator is $\delta$-consistent if
  the sequence of estimates, $\{\hat{X}_t\}_{t=1}^{T}$, computed via $E$, satisfies:
  \begin{equation}
    \Pr[\:\forall t =1 \dots T \colon \mathbf{x}_t \in \hat{X}_t\:] \geq 1-\delta
    \nonumber
  \end{equation}
\end{definition}

This definition states that $\delta$-consistent estimators are able to guarantee that with high probability, jointly for each time step within a finite time horizon, the set-based estimate contains the true state.
Note that:
(1) Def.~\ref{def:delta_consistency} implies that the high-probability consistency guarantee holds \textit{jointly} throughout the finite time horizon, rather than \textit{per time-step} in the form of $\forall t =1 \dots T \colon \Pr[\mathbf{x}_t \in \hat{X}_t] \geq 1-\delta$~\cite{shetty2020predicting}.
(2) Def.~\ref{def:delta_consistency} indicates a filtering (rather than smoothing) problem, as future measurements are never used to estimate past or current states.

\begin{problem}
Design a set-based estimator that guarantees $\delta$-consistency under the epistemic uncertainties in the learned models, $g(\cdot)$ and $h(\cdot)$, and the aleatoric uncertainties in the noises, $\mathbf{w}$ and $\mathbf{v}$.
\end{problem}

\section{Method}\label{sec:method}

Our recursive set-based estimator, GP-ZKF, shown in~\cref{fig:algo_conc}, follows the algorithmic structure of Kalman filter and set-based estimators~\cite{alamo2005guaranteed,le2013zonotope}.
We choose zonotopes to represent the set-based estimates due to their computational advantages in affine transformation and Minkowski sum (see Sec.~\ref{background:zonotope}).

At time step $t=1,\dots,T$, GP-ZKF computes a new zonotopic estimate, $\hat{X}_t \subset \mathbb{R}^{n_x}$ in three phases:
\textit{\textbf{(1) Prediction:}}
Given $\mathbf{u}_{t-1} \in \mathcal{U}$ and $\hat{X}_{t-1} \subset \mathbb{R}^{n_x}$ that contains $\mathbf{x}_{t-1}$ (formally, $\hat{X}_{t-1} \ni \mathbf{x}_{t-1}$), 
GP-ZKF constructs a \textit{dynamics-consistent} zonotope, $\bar{X}_t \subset \mathbb{R}^{n_x}$, to bound the output of the dynamics, $d(\cdot)$, with high probability.
\textit{\textbf{(2) Measurement:}}
Given $\mathbf{u}_{t} \in \mathcal{U}$, $\bar{X}_t$ and a new measurement, $\mathbf{y}_t \in \mathbb{R}^{n_y}$, GP-ZKF constructs a \textit{measurement-consistent} polytope, $\bar{X}_{y_t} \subset \mathbb{R}^{n_x}$, that contains all states consistent with $\mathbf{y}_t$, with high probability.
We choose polytopes to represent the potentially asymmetric set, $\bar{X}_{y_t}$.
\textit{\textbf{(3) Correction:}}
GP-ZKF constructs the new zonotopic estimate, $\hat{X}_t$, via the intersection $\bar{X}_t \cap \bar{X}_{y_t}$.

\subsection{Phase 1: Prediction}\label{ssec:method:prediction}
In the prediction phase, given $\mathbf{u}_{t-1}$ and $\hat{X}_{t-1} \ni \mathbf{x}_{t-1}$, GP-ZKF constructs a \textit{dynamics-consistent} zonotope, $\bar{X}_t$, to bound the outputs of $d(\cdot)$, with high probability.
As defined in~\cref{eq:system:dynamics}, $d(\cdot)$ is decomposed into three components: a known function, $f(\cdot)$; an unknown but to-be-learned function, $g(\cdot)$; and a noise, $\mathbf{w}$.
In this section, we separately formulate bounds for each component and then integrate them into $\hat{X}_{t-1}$ to bound $d(\cdot)$.

\subsubsection{Bounding the Known Function}\label{sssec:method:prediction:known_fcn}
Bounding arbitrary functions is nontrivial.
In this work, we focus on smooth functions by making the following assumption:
\begin{assumption}\label{ass:known_smooth}
(i) $f(\cdot)$ is twice continuously differentiable;
(ii) $f(\cdot)$ is $L_f$-Lipschitz continuous with respect to the $2$-norm; and
(iii) $\| f(\mathbf{x}_0,\mathbf{u}) - \mathbf{x}_0 \|_2 \leq B^f$ for each $\mathbf{x}_0 \in \hat{X}_{0}, \mathbf{u} \in \mathcal{U}$.
\end{assumption}
In this section, we use Assumption~\ref{ass:known_smooth}$(i)$ to bound the outputs of $f(\cdot)$;
we will incorporate Assumption~\ref{ass:known_smooth}$(ii,iii)$ (based on Assumption~3 in \cite{berkenkamp2019safe}) in the following section (Sec.~\ref{sssec:predict:learned_function}).

Assumption~\ref{ass:known_smooth}$(i)$ allows us to directly apply reachability analysis~\cite{althoff2010reachability} to bound the outputs of $f(\cdot)$, given $\mathbf{u}_{t-1}$ and $\hat{X}_{t-1} \ni \mathbf{x}_{t-1}$.
In particular, we first linearize $f(\cdot)$ around a reference point $\bar{\mathbf{x}}_{t-1} \in \mathbb{R}^{n_x}$, which is chosen to be the center of $\hat{X}_{t-1}$.
The linearized function $\bar{f}(\mathbf{x}_{t-1}, \mathbf{u}_{t-1}) = f(\bar{\mathbf{x}}_{t-1}, \mathbf{u}_{t-1}) + J^f_x \cdot (\mathbf{x}_{t-1} - \bar{\mathbf{x}}_{t-1})$, where $J^f_x$ is the Jacobian of $f(\cdot)$ with respect to $\mathbf{x}_{t-1}$, evaluated at $(\bar{\mathbf{x}}_{t-1}, \mathbf{u}_{t-1})$.
Given $\mathbf{u}_{t-1}$ and
$\hat{X}_{t-1} \ni \mathbf{x}_{t-1}$,
the linearization error can be bounded by a zero-centered box that $\subset \mathbb{R}^{n_x}$ (see Prop.~3.7~\cite{althoff2010reachability}).
Formally, for all $\mathbf{x}_{t-1} \in \hat{X}_{t-1}$, $\mathbf{u}_{t-1} \in \mathcal{U}$, we have:
\begin{equation}
  f(\mathbf{x}_{t-1}, \mathbf{u}_{t-1}) - \bar{f}(\mathbf{x}_{t-1}, \mathbf{u}_{t-1}) \in R^f(\hat{X}_{t-1}, \mathbf{u}_{t-1})
  \subset \mathbb{R}^{n_x}
\label{eq:known_fcn_error}
\end{equation}
Here, $R^f(\cdot)$ denotes the function used to compute the box that bounds the error based on $\hat{X}_{t-1}$ and $\mathbf{u}_{t-1}$.

\subsubsection{Bounding the Unknown Function}\label{sssec:predict:learned_function}
The unknown function $g(\cdot)$ is learned via a GP (see Sec.~\ref{background:gp}).
In this section, we formulate a high-probability bound for the output of $g(\mathbf{x}_{t-1},\mathbf{u}_{t-1})$, given $\mathbf{u}_{t-1}$ and $\hat{X}_{t-1} \ni \mathbf{x}_{t-1}$, in the following five steps:
(i) We state certain regularization assumptions about $g(\cdot)$;
(ii) show that the state space is compact, with high probability;
(iii) utilize this compactness to derive a bound for the GP posterior mean, $\mu^g(\cdot)$;
(iv) present a bound for the GP posterior standard deviation, $\sigma^g(\cdot)$;
and
(v) integrate both bounds of $\mu^g(\cdot)$ and $\sigma^g(\cdot)$ into the GP confidence intervals, as introduced in Lemma~\ref{lemma:well_calibrated}, to bound the output of $g(\cdot)$.

\underline{(i) \textit{Regularization assumptions}}:
We denote the kernel for $g(\cdot)$ as $k^g$ and the single-output surrogate function as $g'(\cdot)$ (see Sec.~\ref{background:gp}),
and assume the following:
\begin{assumption}\label{ass:unknown_smooth}
(i) $k^g$ is 2-times continuously differentiable (Def.~4.35 in \cite{steinwart2008support}),
(ii) $k^g$ is bounded by $\|k^g\|_\infty$ (Eq.~4.15 in \cite{steinwart2008support}),
(iii) $k^g$ has bounded derivatives (Assumption~4 in \cite{berkenkamp2019safe}),
and
(iv) $\|g'\|_{k^g} \leq B^g$ (see Sec.~\ref{background:gp}).
\end{assumption}
This assumption states that $k^g(\cdot,\cdot)$ and $g(\cdot)$ are smooth and bounded;
common smooth kernels, such as square exponential and rational quadratic kernels, satisfy this assumption.
Assumption \ref{ass:unknown_smooth} $(i,iii,iv)$ implies that $g(\cdot)$ is $L_g$-Lipschitz continuous with respect to the $2$-norm (Cor.~2 in \cite{berkenkamp2019safe}).

\underline{(ii) \textit{Compact state space}}:
The infinite support of the process noise, $\mathbf{w}$, makes the state space of our system unbounded.
Taking a step back, the following lemma states that, with high probability, the \textit{reachable} state space is compact during the estimation processes:

\begin{lemma}\label{lemma:compact_state_space}
  With Asm.~\ref{ass:known_smooth} $(ii,iii)$,
  Asm.~\ref{ass:unknown_smooth},
  and our Gaussian Asm. on $\mathbf{w}$ (Sec.~\ref{sec:system_formulation}),
  and
  given initial set $\hat{X}_0 \ni \mathbf{x}_0$,
  then with probability at least $(1-\delta^w)$,
  jointly
  for all $t = 0,\dots,T$,
  we have that $\mathbf{x}_t \in \mathcal{X}$,
  where $\mathcal{X}$ is a box $\subset \mathbb{R}^{n_x}$ (compact) that depends upon $\hat{X}_0$, $T$, $\lambda^w$, $\delta^w$,
  $L^f$, $L^g$, $B^f$, $B^g$, $\|k^g\|_{\infty}$, $n_x$.
\end{lemma}

Lemma~\ref{lemma:compact_state_space} can be proven by integrating $\hat{X}_0$ and our noise bound, $W$ (\cref{remark:noise_bounds} in Sec.~\ref{sec:system_formulation}), into Lemma 44 in \cite{berkenkamp2019safe}.
The intuition behind the proof is that our Assumption~\ref{ass:known_smooth} $(ii,iii)$, Assumption~\ref{ass:unknown_smooth}, and the Gaussian nature of $\mathbf{w}$ (with a probability budget $\delta^w$) ``prevent'' $f(\cdot)$, $g(\cdot)$, and $\mathbf{w}$ (respectively) from deviating the state infinitely far away from $\hat{X}_0$.

\underline{(iii) \textit{Bound posterior mean}}:
Next, we derive a bound for the posterior mean, $\mu^g(\cdot)$.
We first linearize $\mu^g(\cdot)$ around $\bar{\mathbf{x}}_{t-1}$, which is chosen to be the center of $\hat{X}_{t-1}$.
The linearized function is $\bar{\mu}^g(\mathbf{x}_{t-1}, \mathbf{u}_{t-1}) = \mu^g(\bar{\mathbf{x}}_{t-1}, \mathbf{u}_{t-1}) + J^{\mu_g}_x \cdot (\mathbf{x}_{t-1} - \bar{\mathbf{x}}_{t-1})$, where $J^{\mu_g}_x$ is the Jacobian of $\mu^g$.
Lemma~\ref{lemma:compact_state_space} implies that the domain of $\mu^g(\cdot)$ during the estimation process is compact, with high probability.
Together with Assumption~\ref{ass:unknown_smooth} $(i)$, we obtain that, with high probability, for each $j=1,\dots,n_x$, the mean $\mu^g_j$ is twice continuously differentiable with $L^g_{\nabla \mu}$-Lipschitz gradient.
We then follow the steps in V(A)2) in~\cite{koller2019arXivlearning} to derive a bound for the linearization error, as follows:
\begin{equation}
\hspace{-3mm}
\resizebox{.94\hsize}{!}{$
  |\mu_j^g(\mathbf{x}_{t-1}, \mathbf{u}_{t-1}) - \bar{\mu}_j^g(\mathbf{x}_{t-1}, \mathbf{u}_{t-1})| \leq \frac{1}{2} L^g_{\nabla \mu} \cdot \|\mathbf{x}_{t-1} - \bar{\mathbf{x}}_{t-1}\|_2^2
  \label{eq:predict:learned_function:mean}
$}
\end{equation}
%
Here, with probability at least $(1-\delta^w)$, this bound holds uniformly for all $t=1,\dots,T$, $j=1,\dots,n_x$, $\mathbf{x}_{t-1} \in \mathcal{X}$, $\mathbf{u}_{t-1} \in \mathcal{U}$, where the probability $(1-\delta^w)$ is derived from Lemma~\ref{lemma:compact_state_space}.

\underline{(iv) \textit{Bound standard deviation}}:
We approximate the standard deviation, $\sigma^g(\mathbf{x}_{t-1},\mathbf{u}_{t-1})$, using $\sigma^g(\bar{\mathbf{x}}_{t-1},\mathbf{u}_{t-1})$.
Assumption~\ref{ass:unknown_smooth} $(i, ii, iii)$ allows the use of (Eq.~21,22 in \cite{lederer2019uniform}) to bound the approximation error.
Formally, with $L^g_{\sigma} \in \mathbb{R}$,
we have for all $t=1,\dots,T$, $j=1,\dots,n_x$, $\mathbf{x}_{t-1} \in \mathbb{R}^{n_x}$, $\mathbf{u}_{t-1} \in \mathcal{U}$:
\begin{equation}
\hspace{-2.25mm}
\resizebox{.93\hsize}{!}{$
  |\sigma^g_j(\mathbf{x}_{t-1},\mathbf{u}_{t-1}) - \sigma^g_j(\bar{\mathbf{x}}_{t-1},\mathbf{u}_{t-1})| \leq L^g_{\sigma} \cdot \|\mathbf{x}_{t-1} - \bar{\mathbf{x}}_{t-1}\|_2^{1/2}
  \label{eq:predict:learned_function:var}
$}
\end{equation}

\underline{(v) \textit{Bound combination}}:
Assumption~\ref{ass:unknown_smooth} $(iv)$ allows the use of Lemma~\ref{lemma:well_calibrated} to construct confidence intervals for $g(\cdot)$ with $\delta^g \in (0,1)$.
Via a union bound, we combine the confidence intervals and our bounds for the mean and standard deviation to bound the error, $|g(\cdot)-\bar{\mu}^g(\cdot)|$, as follows:
\begin{equation}
\hspace{-3.75mm}
\begin{aligned}
  & \hspace{2mm}|g_j(\mathbf{x}_{t-1}, \mathbf{u}_{t-1}) - \bar{\mu}^g_j(\mathbf{x}_{t-1}, \mathbf{u}_{t-1})|\\
  & \hspace{3mm} \leq \hspace{-6mm}
  \underbrace{\frac{1}{2} L^g_{\nabla \mu} \cdot \epsilon^2}_\text{Linearization error of $\mu^g_j(\cdot)$ } \hspace{-2mm}+ 
  \underbrace{\frac{}{}\beta^g \cdot \sigma^g_{j}(\bar{\mathbf{x}}_{t-1}, \mathbf{u}_{t-1})
  + \beta^g \cdot L^g_{\sigma} \cdot \epsilon^{1/2}}_\text{Epistemic uncertainty $\bigoplus$ Approx. error of $\sigma^g_j(\cdot)$
  }
  \hspace{-1mm}
  \end{aligned}
\label{eq:predict:learned_function:confidence_bounds}
\end{equation}
Here, $\epsilon = \|\hat{X}_{t-1}-\bar{\mathbf{x}}_{t-1}\|_2$ is the norm of the translated zonotope.
(We provide the derivation in the Appendix.)
Let $R^g(\hat{X}_{t-1}, \mathbf{u}_{t-1})$ be a zero-centered box $\subset \mathbb{R}^{n_x}$, whose radius along each dimension $j=1,\dots,n_x$ is the RHS of Eq.~\ref{eq:predict:learned_function:confidence_bounds}.
Thus, with probability at least $(1-\delta^g-\delta^w)$, jointly for all $t=1,\dots,T$, $\hat{X}_{t-1} \subset \mathcal{X}$, $\mathbf{u}_{t-1} \in \mathcal{U}$, $\mathbf{x}_{t-1} \in \hat{X}_{t-1}$, we have that
\begin{equation}
  g(\mathbf{x}_{t-1}, \mathbf{u}_{t-1}) - \bar{\mu}^g(\mathbf{x}_{t-1}, \mathbf{u}_{t-1}) \in R^g(\hat{X}_{t-1}, \mathbf{u}_{t-1})
\label{eq:predict:learned_function:bound}
\end{equation}

\subsubsection{The Prediction Phase}\label{sssec:predict:predict}
We define a linear function, $\bar{d}(\cdot)$, by summing the linearized $f(\cdot)$ and $\mu_g(\cdot)$ as $\bar{d}(\mathbf{x}_{t-1}, \mathbf{u}_{t-1}) = \bar{f}(\mathbf{x}_{t-1}, \mathbf{u}_{t-1}) + \bar{\mu}^g(\mathbf{x}_{t-1}, \mathbf{u}_{t-1})$.
We define a box, $R^d(\hat{X}_{t-1}, \mathbf{u}_{t-1}) \subset \mathbb{R}^{n_x}$, by Minkowski summing the error bounds of
$f(\cdot)$ (Eq.~\ref{eq:known_fcn_error}),
$g(\cdot)$ (Eq.~\ref{eq:predict:learned_function:bound}),
and $w$ (\cref{remark:noise_bounds} in Sec.~\ref{sec:system_formulation}).
Then, with probability at least $(1-\delta^g-\delta^w)$, the error between $d(\cdot)$ and $\bar{d}(\cdot)$ can be bounded by $R^d(\cdot)$ as follows: 
\begin{equation}
\begin{aligned}
  &d(\mathbf{x}_{t-1}, \mathbf{u}_{t-1}, \mathbf{w}_{t-1}) - \bar{d}(\mathbf{x}_{t-1}, \mathbf{u}_{t-1})
  \in
  R^d(\hat{X}_{t-1}, \mathbf{u}_{t-1})
  \\
  & \coloneqq \hspace{-2mm}
  \underbrace{R^f(\hat{X}_{t-1}, \mathbf{u}_{t-1})}_{\text{Linearization err. of } f(\cdot)}
  \hspace{-1mm}
  \bigoplus
  \hspace{-1mm}
  \underbrace{R^g(\hat{X}_{t-1}, \mathbf{u}_{t-1})}_{
    \substack{\text{Epistemic } \bigoplus \text{ Lin. err. of } \mu^g(\cdot)\\
              \bigoplus \text{Approx. err. of } \sigma^g(\cdot)}}
  \bigoplus
  \underbrace{W}_{\text{Aleatoric}}
\end{aligned}
\label{eq:predict:bounds}
\vspace{-1mm}
\end{equation}
Given that $\bar{d}$ is a linear function, $\mathbf{u}_{t-1}$ is known, and $\mathbf{x}_{t-1} \in \hat{X}_{t-1}$, the output $\bar{d}(\mathbf{x}_{t-1}, \mathbf{u}_{t-1})$ is bounded by a zonotope obtained by linearly transforming $\hat{X}_{t-1}$.
Then, by Minkowski-summing this transformed zonotope and $R^d(\hat{X}_{t-1}, \mathbf{u}_{t-1})$, we obtain another zonotope, denoted by $\bar{X}_t$, that bounds the output $d(\mathbf{x}_{t-1}, \mathbf{u}_{t-1}, \mathbf{w}_{t-1})$ (see Sec.~\ref{background:zonotope}).
Let $D(\hat{X}_{t-1}, \mathbf{u}_{t-1})$ denote the function to compute $\bar{X}_t$;
we now summarize the steps described throughout this section in the following lemma:
\begin{lemma}\label{lemma:prediction_full}
  Let $\delta^g, \delta^w \in (0,1)$.
  For all $n \in \mathbb{N}$, we choose $\beta_n^g$ according to Lemma~\ref{lemma:well_calibrated}.
  Given an initial set $\hat{X}_{n,0} \ni \mathbf{x}_{n,0}$,
  then, with probability at least $(1-\delta^g-\delta^w)$,
  the following holds jointly for all $n \in \mathbb{N}$, $t=1,\dots,T$:
  \begin{itemize}
    \item[$(i)$] $d(\mathbf{x}_{n,t-1}, \mathbf{u}_{n,t-1}, \mathbf{w}_{n,t-1}) \in \bar{X}_{n,t} = D(\hat{X}_{n,t-1},\mathbf{u}_{n,t-1})$,\\
    for all $\hat{X}_{n,t-1} \subset \mathcal{X}$, $\mathbf{x}_{n,t-1} \in \hat{X}_{n,t-1}$, $\mathbf{u}_{n,t-1} \in \mathcal{U}$,\\
    where $\mathbf{w}_{n,t-1}$ is the process noise as assumed in Sec.~\ref{sec:system_formulation}.
    \item[$(ii)$] $\mathbf{x}_{n,t} \in \mathcal{X}$, $\mathbf{x}_{n,0} \in \mathcal{X}$.
  \end{itemize}
\end{lemma}
\begin{proof}
The noise bound, $W$ (\cref{remark:noise_bounds}), and Lemma~\ref{lemma:compact_state_space} hold jointly, with probability at least $(1-\delta^w)$.
Then, by applying a union bound to combine the above result with the confidence intervals of $g(\cdot)$ (Lemma~\ref{lemma:well_calibrated}), we arrive at the bound, $\bar{X}_{n,t}$.
\end{proof}
In conclusion, the prediction phase computes the dynamic-consistent zonotope, $\bar{X}_t = D(\hat{X}_{t-1},\mathbf{u}_{t-1}) \subset \mathbb{R}^{n_x}$, which bounds the output of the dynamic function, $d(\cdot)$ (Lemma~\ref{lemma:prediction_full}).

\subsection{Phase 2: Measurement}\label{ssec:method:measurement}
In the measurement phase, 
given $\bar{X}_t$ from Lemma~\ref{lemma:prediction_full} and $\mathbf{u}_{t}$,
GP-ZKF constructs a \textit{measurement-consistent} polytope~\cite{le2013zonotope}, $\bar{X}_{y_t}$, to bound all possible states consistent with the new measurement, $\mathbf{y}_t$, with high probability.

The unknown function $h(\cdot)$ is learned via GP (see Sec.~\ref{background:gp}).
We denote the single-output surrogate function by $h'(\cdot)$, and the kernel by $k^h$.
Similar to Assumption~\ref{ass:unknown_smooth} about $k^g$ and $g(\cdot)$, we make the following assumption about $k^h$ and $h(\cdot)$:
\begin{assumption}\label{ass:obs:unknown_smooth}
(i) $k^h$ is 2-times continuously differentiable (Def.~4.35 in \cite{steinwart2008support}),
(ii) $k^h$ is bounded by $\|k^h\|_\infty$ (Eq.~4.15 in \cite{steinwart2008support}),
(iii) $k^h$ has bounded derivatives (Assumption~4 in \cite{berkenkamp2019safe}),
and
(iv) $\|h'\|_{k^h} \leq B^h$ (see Sec.~\ref{background:gp}).
\end{assumption}

To bound $h(\cdot)$, we follow the same steps described in Sec.~\ref{sssec:predict:learned_function} to bound $g(\cdot)$.
\underline{\textit{Bound posterior mean}}:
We linearize $\mu^h(\cdot)$ around $\bar{\mathbf{x}}_{t} =$ the center of $\bar{X}_{t}$, and obtain $\bar{\mu}^h(\mathbf{x}_{t}, \mathbf{u}_{t}) = \mu^h(\bar{\mathbf{x}}_{t}, \mathbf{u}_{t}) + J^{\mu_h}_x \cdot (\mathbf{x}_{t} - \bar{\mathbf{x}}_{t})$.
Lemma~\ref{lemma:prediction_full} $(ii)$ implies the compactness of the domain of $\mu^h(\cdot)$ within the estimation process, which leads to a linearization error bound (similar to Eq.~\ref{eq:predict:learned_function:mean}).
\underline{\textit{Bound standard deviation}}:
With Assumption~\ref{ass:obs:unknown_smooth}, we derive a bound for $\sigma^h(\cdot)$ (similar to Eq.~\ref{eq:predict:learned_function:var}).
\underline{\textit{Bound combination}}:
Given $\delta^h \in (0,1)$, Lemma~\ref{lemma:well_calibrated} allows us to construct confidence intervals for $h(\cdot)$.
Via a union bound, we determine that the confidence intervals and noise bound $\mathbf{v_t} \in V$ (see ~\cref{remark:noise_bounds}) jointly hold with probability at least $(1-\delta^h-\delta^v)$.
Then, similar to $R^g$ in Eq.~\ref{eq:predict:learned_function:bound}, we obtain a box, $R^h(\bar{X}_{t}, \mathbf{u}_{t}) \subset \mathbb{R}^{n_y}$.
Thus, with probability at least $(1-\delta^g-\delta^w)(1-\delta^h-\delta^v)$,
jointly for all $t=1,\dots,T$, $\bar{X}_{t} \subset \mathcal{X}$, $\mathbf{u}_{t} \in \mathcal{U}$, $\mathbf{x}_{t} \in \bar{X}_{t}$, we have that
\begin{equation}
  h(\mathbf{x}_{t}, \mathbf{u}_{t}) - \bar{\mu}^h(\mathbf{x}_{t}, \mathbf{u}_{t}) \in R^h(\bar{X}_{t}, \mathbf{u}_{t})
\label{eq:measurement:learned_function:bound}
\end{equation}
Here, the product rule, $(1-\delta^g-\delta^w)(1-\delta^h-\delta^v)$, results from the assumption that noises $\mathbf{w}$ and $\mathbf{v}$ are independent with each other (Sec.~\ref{sec:system_formulation}).
By expanding $\bar{\mu}^h$ and then combining the noise bound $V$ (\cref{remark:noise_bounds}) with Eq.~\ref{eq:measurement:learned_function:bound}, we obtain that
with probability at least $(1-\delta^g-\delta^w)(1-\delta^h-\delta^v)$,
jointly for all $t=1,\dots,T$, $\bar{X}_{t} \subset \mathcal{X}$, $\mathbf{u}_{t} \in \mathcal{U}$, $\mathbf{x}_{t} \in \bar{X}_{t}$,
the following holds:
\begin{equation}
\begin{aligned}
  &\bar{\mu}^h(\mathbf{x}_{t}, \mathbf{u}_{t}) - o(\mathbf{x}_{t}, \mathbf{u}_{t}, \mathbf{v}_{t})
  \\
  & = \mu^h(\bar{\mathbf{x}}_{t}, \mathbf{u}_{t}) + J^{\mu_h}_x \cdot (\mathbf{x}_{t} - \bar{\mathbf{x}}_{t})
  - h(\mathbf{x}_{t}, \mathbf{u}_{t}) - \mathbf{v}_{t}
  \\
  &\quad\in R^o(\bar{X}_{t}, \mathbf{u}_{t})
  \coloneqq
  \underbrace{R^h(\bar{X}_{t}, \mathbf{u}_{t})}_{
    \substack{\text{Epistemic } \bigoplus \text{ Lin. err. of } \mu^h(\cdot)\\
            \bigoplus \text{Approx. err. of } \sigma^h(\cdot)}
    }
  \bigoplus
  \underbrace{V}_{\text{Aleatoric}}
%
\end{aligned}
\label{eq:measurement:bounds}
\end{equation}

Given a new measurement, $\mathbf{y}_t = o(\mathbf{x}_{t}, \mathbf{u}_{t}, \mathbf{v}_{t}) \in \mathbb{R}^{n_y}$ (Eq.~\ref{eq:system:measurement}), the measurement-consistent states must satisfy the bound in Eq.~\ref{eq:measurement:bounds}.
Hence, we equivalently represent the bound in Eq.~\ref{eq:measurement:bounds} as a polytope, $\bar{X}_{y_t} \subset \mathbb{R}^{n_x}$, with $\mathbf{x}_{t}$ as the variable, as follows:
\begin{equation}
\hspace{-2mm}
\resizebox{.94\hsize}{!}{$
\begin{aligned}
  \bar{X}_{y_t} =
  \{\mathbf{x}_t \in \mathbb{R}^{n_x} \colon
  &J^{\mu_h}_x \cdot \mathbf{x}_{t} - [\mathbf{y}_t - \mu^h(\bar{\mathbf{x}}_{t}, \mathbf{u}_{t}) + J^{\mu_h}_x \cdot \bar{\mathbf{x}}_{t}] \in R^o(\bar{X}_{t}, \mathbf{u}_{t})
  \}
%
\end{aligned}
$}
\hspace{-3mm}
\label{eq:measurement:polytope}
\end{equation}

The output of the measurement step is the measurement-consistent polytope, $\bar{X}_{y_t} \subset \mathbb{R}^{n_x}$.
Let $O^{inv}(\bar{X}_{t}, \mathbf{u}_t, \mathbf{y}_t)$ denote the function to compute $\bar{X}_{y_t}$, where the superscript $inv$ emphasizes that $O^{inv}$ is the ``inverse'' of our observation model, $o(\cdot)$.

\subsection{Phase 3: Correction}\label{ssec:method:correction}
In the correction phase, the goal is to construct a zonotope $\hat{X}_t$ as the intersection $\bar{X}_t \cap \bar{X}_{y_t}$, where $\bar{X}_t$ is formulated in Lemma~\ref{lemma:prediction_full} and $\bar{X}_{y_t}$ is defined in Eq.~\ref{eq:measurement:polytope}.
Note that the bound in Eq.~\ref{eq:measurement:bounds} requires $\bar{X}_{t} \subset \mathcal{X}$; therefore, GP-ZKF instead computes the intersection $\bar{X}_t \cap \bar{X}_{y_t} \cap \mathcal{X}$.
This intersection cannot be computed exactly, and hence is outer-approximated by $\hat{X}_t$.
The new zonotope, $\hat{X}_t$, is potentially less conservative than $\bar{X}_{t}$ and $\bar{X}_{y_t}$ in bounding the possible states, as $\hat{X}_t$ takes into account both dynamics and measurements.
Formally:
%
\begin{lemma}\label{lemma:correction_full}
  Let $\delta^g, \delta^w, \delta^h, \delta^v \in (0,1)$.
  For all $n^g,n^h \in \mathbb{N}$, we choose $\beta_{n^g}^g, \beta_{n^h}^h$ according to Lemma~\ref{lemma:well_calibrated}\footnote{We omit the subscripts $n^g,n^h$ for every variable in this lemma for clarity.}.
  Given an initial set $\hat{X}_{0} \ni \mathbf{x}_{0}$,
  then, with probability at least $(1-\delta^g-\delta^w)(1-\delta^h-\delta^v)$,
  jointly for all $n^g,n^h \in \mathbb{N}$,
  $t=1,\dots,T$,
  $\hat{X}_{t-1} \subset \mathcal{X}$,
  $\mathbf{x}_{t-1} \in \hat{X}_{t-1}$,
  $\mathbf{u}_{t-1}, \mathbf{u}_{t} \in \mathcal{U}$,
  and $\mathbf{y}_{t} \in \mathbb{R}^{n_y}$,
  we have the following:
  \begin{equation}
      d(\mathbf{x}_{t-1}, \mathbf{u}_{t-1}, \mathbf{w}_{t-1})
      \in
      \left(\bar{X}_t \cap \bar{X}_{y_t} \cap \mathcal{X}\right)
      \subset \hat{X}_{t}
    \label{eq:lemma:correction_full}
  \end{equation}
  where $\mathbf{w}_{n,t-1}$ is the process noise as assumed in Sec.~\ref{sec:system_formulation},
  $\bar{X}_{t} = D(\hat{X}_{t-1},\mathbf{u}_{t-1})$ (Lemma~\ref{lemma:prediction_full}), and
  $\bar{X}_{y_t} = O^{inv}(\bar{X}_{t}, \mathbf{u}_t, \mathbf{y}_t)$ (Eq.~\ref{eq:measurement:polytope}).
\end{lemma}

This lemma states that the true state at time $t$ lies within the zonotope, $\hat{X}_{t}$, with high probability.
Lemma~\ref{lemma:correction_full} summarizes the derivations in Sec.~\ref{ssec:method:measurement};
it can be proved by directly combining Lemma~\ref{lemma:prediction_full} and the bound in Eq.~\ref{eq:measurement:bounds}.

The intersection in Eq.~\ref{eq:lemma:correction_full} is outer-approximated in two steps:
First, GP-ZKF follows Prop.~1 in~\cite{le2013zonotope} to obtain a zonotope denoted by $Z_t(\Lambda_t)$ and parameterized by the matrix $\Lambda_t$, such that $Z_t(\Lambda_t) \supset \left( \bar{X}_t \cap \bar{X}_{y_t}\right)$.
The parameter $\Lambda_t$ is obtained by analytically solving a convex program that minimizes the ``size'' of $Z_t(\Lambda_t)$ (see Sec.~6.1 in~\cite{alamo2005guaranteed}).
Then, GP-ZKF follows the same procedures to construct $\hat{X}_{t} \supset \left(\bar{X}_t \cap \bar{X}_{y_t} \cap \mathcal{X}\right)$.

The essence of the correction step can be represented by the function $E(\hat{X}_{t-1},\mathbf{u}_{t-1},\mathbf{u}_t,\mathbf{y}_t)$, which conducts the intersections and outputs the new zonotopic estimate, $\hat{X}_{t}$.

\section{Theoretical Guarantees}\label{sec:theory}
We present two key theorems about GP-ZKF.
First, as both epistemic and aleatoric uncertainties have been bounded, GP-ZKF is $\delta$-consistent (Def.~\ref{def:delta_consistency});
second, when both uncertainties are relaxed, GP-ZKF reduces to GP-EKF~\cite{ko2009gp}.
\begin{theorem}\label{thm:consistent}
  Given $\delta \in (0, 1)$, GP-ZKF chooses $\delta^g, \delta^h, \delta^w, \delta^v \in (0,1)$ such that $(1-\delta^g-\delta^w)(1-\delta^h-\delta^v) \geq (1-\delta)$.
  For all $n^g,n^h \in \mathbb{N}$, GP-ZKF chooses $\beta_{n^g}^g, \beta_{n^h}^h$ according to Lemma~\ref{lemma:well_calibrated}.
  Then, GP-ZKF is $\delta$-consistent.
\end{theorem}

\begin{proof}
Similar to the proof for Cor.~7 in \cite{koller2019arXivlearning}, we recursively apply Lemma~\ref{lemma:correction_full} from $t=1$ to $T$ and obtain with high probability, jointly for all $t=1,\dots,T$, that $\mathbf{x}_t \in \hat{X}_t$.
\end{proof}

GP-EKF~\cite{ko2009gp} is an EKF-based state estimator that learns both dynamic and observation models via GPs;
we see GP-EKF as the \textit{stochastic} counterpart to our \textit{set-based} GP-ZKF.
At every time $t=1,\dots,T$, GP-EKF computes the Kalman gain, $K_t \in \mathbb{R}^{n_x \times n_y}$, and outputs a point estimate, $\mu_t \in \mathbb{R}^{n_x}$, and a covariance, $\Sigma_t \in \mathbb{R}^{n_x \times n_x}$ (see Table~2 in \cite{ko2009gp}).
To draw an analogy similar to Thm.7 in~\cite{combastel2015zonotopes}, we define GP-ZKF's Kalman gain as the matrix $\Lambda_t$, which parameterizes the intersection $\bar{X}_t \cap \bar{X}_{y_t}$, as mentioned at the end of Sec.~\ref{ssec:method:correction}.
We define GP-ZKF's point estimate as the zonotope center, $(\hat{X}_t)_{c}$, and covariance as $(\hat{X}_t)_{G}((\hat{X}_t)_{G})^T$ (the zonotope covariation, as seen in Def.~4 in \cite{combastel2015zonotopes}).
Next, we formally demonstrate that GP-ZKF can be reduced to (the Joseph form of) GP-EKF:
\begin{theorem}\label{thm:equal_to_gpekf}
  Assume that GP-ZKF's set-based estimates are always inside the state space---i.e., for all $t=1,\dots,T$, the zonotope $Z_t(\Lambda_t) \subset \mathcal{X}$, where $Z_t(\Lambda_t)$ is defined at the end of Sec.~\ref{ssec:method:correction}.
  Given the same initial condition, $\mu_0=(\hat{X}_0)_{c}$ and $K_0 = (\hat{X}_0)_{G}((\hat{X}_0)_{G})^T$,
  if GP-ZKF:
  $(i)$ sets the confidence interval scalings for $g(\cdot)$ and $h(\cdot)$, denoted by $\beta^g$ and $\beta^h$ (respectively), to $1$;
  $(ii)$ omits all noise bounds for $\mathbf{w}$ and $\mathbf{v}$ by setting $W=\emptyset$ and $V=\emptyset$ (Remark~\ref{remark:noise_bounds}); and
  $(iii)$ omits all linearization errors for $f(\cdot)$, $\mu^g$, $\sigma^g$, $\mu^h$, and $\sigma^h$ by setting $R^f(\cdot)=\emptyset$ (Eq.~\ref{eq:known_fcn_error}),
  $L^g_{\nabla \mu}=0$ (Eq.~\ref{eq:predict:learned_function:mean}),
  $L^g_{\sigma}=0$ (Eq.~\ref{eq:predict:learned_function:var}),
  $L^h_{\nabla \mu}=0$ (Sec.~\ref{ssec:method:measurement}),
  and
  $L^h_{\sigma}=0$ (Sec.~\ref{ssec:method:measurement});
  then, for time $t=1,\dots,T$, 
  we have $K_t = \Lambda_t$, $\mu_t = (\hat{X}_t)_{c}$, and $\Sigma_t = (\hat{X}_t)_{G}((\hat{X}_t)_{G})^T$.
\end{theorem}
\begin{proof}
%
%
%
Under the relaxations above,
Eq.~\ref{eq:lemma:correction_full} becomes $\left( \bar{X}_t \cap \bar{X}_{y_t} \cap \mathcal{X}\right) \subset Z_t(\Lambda_t) = \hat{X}_{t}$.
And each of the uncertainty bounds, $R^d(\cdot)$ (Eq.~\ref{eq:predict:bounds}) and $R^o(\cdot)$ (Eq.~\ref{eq:measurement:bounds}), only contains one standard deviation.
%
%
%
As introduced at the end of Sec.~\ref{ssec:method:correction}, GP-ZKF obtains $\hat{X}_{t}(\Lambda_t)$ by optimizing $\Lambda_t$.
%
With the analytical solution, $\Lambda_t$, we reach the final conclusion by induction (see the proof of Thm.7 in~\cite{combastel2015zonotopes}).
%
%
\end{proof}

This theorem states that with certain relaxations, GP-ZKF could produce the same Kalman gain, point estimate, and covariance as GP-EKF.
The Kalman gain in GP-ZKF, $\Lambda_t$, weighs the dynamic-consistent zonotope, $\bar{X}_t$, and the measurement-consistent polytope, $\bar{X}_{y_t}$, when ``mixing'' them within the outer-approximated intersection.
In contrast to GP-EKF, Thm.~\ref{thm:equal_to_gpekf} signifies the conservativeness of GP-ZKF in bounding the linearization errors and aleatoric and epistemic uncertainties during estimation.
The conservativeness echoes GP-ZKF's consistency guarantee, as stated in Thm.~\ref{thm:consistent}.

\section{Experiment and Result}\label{sec:experiment_and_result}

\newcolumntype{R}{>{\centering\arraybackslash}p{1.975cm}}
\newcolumntype{S}{>{\centering\arraybackslash}p{0.3cm}}
\newcolumntype{O}{>{\centering\arraybackslash}p{0.575cm}}
\newcolumntype{Y}{>{\centering\arraybackslash}X}
\newcolumntype{Z}{>{\centering\arraybackslash}p{0.85cm}}
\newcolumntype{J}{>{\centering\arraybackslash}p{0.7cm}}
\newcolumntype{K}{>{\centering\arraybackslash}p{0.525cm}}
\newcolumntype{L}{>{\centering\arraybackslash}p{0.985cm}}
\newcolumntype{I}{>{\centering\arraybackslash}p{1.575cm}}
\newcolumntype{E}{>{\centering\arraybackslash}p{1.425cm}}
\newcolumntype{Q}{>{\centering\arraybackslash}p{1.2cm}}
\newcolumntype{M}{>{\centering\arraybackslash}p{1.1cm}}
\begin{table}[t]
    \centering
    \vspace{2mm}
    \caption{The estimation results in the Simulated Pendulum Domain.}
    \label{tab:num_pendulum}
    \begin{tabular}{MMISIK}
        \toprule
        Data & Method & RMSE $(\theta,\dot{\theta})$ & Incl. & Radius $(\theta,\dot{\theta})$  & Time \\
        shift &  &  
        ( $ ^\circ$, $^\circ$/$s$) & (\%) & ( $ ^\circ$, $^\circ$/$s$) & (s)\\
        \midrule
        Both
& GP-EKF & 20.2, 37.5 & 5 & 1.7, 9.2 & 0.003\\
models
& GP-UKF & 7.4, 15.6 & 38 & 10.0, 32.1 & 0.009\\
& GP-PF & 90.5, 64.4 & 28 & 27.0, 54.6 & 1.459\\
& \textbf{GP-ZKF} & 16.4, 20.4 & \textbf{83} & 46.7, 131.4 & 0.004\\
        \midrule
        Dynamics
& GP-EKF & 15.3, 22.1 & 0.17 & 0.2, 0.4 & 0.003\\
model & GP-UKF & 2.2, 10.8 & 27 & 1.6, 17.8 & 0.009\\
& GP-PF & 69.6, 46.9 & 16 & 11.8, 25.5 & 1.500\\
& \textbf{GP-ZKF} & 16.4, 21.5 & \textbf{88} & 44.6, 72.5 & 0.004\\
        \midrule
        Observation
& GP-EKF & 15.6, 21.7 & 4 & 1.4, 2.0 & 0.003\\
model & GP-UKF & 25.4, 31.2 & 27 & 35.9, 40.6 & 0.009\\
& GP-PF & 222.7, 157.4 & 18 & 82.9, 100.2 & 1.544\\
& \textbf{GP-ZKF} & 15.3, 18.6 & \textbf{87} & 47.9, 126.0 & 0.004\\
        \midrule
        None
& GP-EKF & 15.4, 20.4 & 0.00 & 0.2, 0.4 & 0.003\\
& GP-UKF & 2.2, 7.8 & 18 & 1.5, 5.1 & 0.009\\
& GP-PF & 161.3, 126.7 & 10 & 47.0, 47.8 & 1.573\\
& \textbf{GP-ZKF} & 16.6, 19.7 & \textbf{92} & 47.1, 72.8 & 0.004\\
        \bottomrule
    \end{tabular}
    \vspace{-4mm}
\end{table}

In Sec.~\ref{sec:theory}, we theoretically highlight GP-ZKF's consistency guarantee and its connection with GP-EKF.
We also empirically demonstrate GP-ZKF's improved consistency against three GP-based stochastic approaches: GP-EKF, GP-UKF, and GP-PF~\cite{ko2009gp}, in both a simulated inverted pendulum domain under significant epistemic uncertainty and a real-world robot-assisted dressing domain.

In both domains, we evaluated all methods using the following metrics:
(1) Avg. root-mean-square error (RMSE) (per dimension) for the point estimate.
GP-ZKF's point estimate is defined as the zonotope center, $(\hat{X}_t)_{c}$ (see Sec.~\ref{sec:theory}).
(2) Inclusion: the percentage of time steps during which the true state is ``included'' by the set-based estimate, which measures a method's consistency.
Recall that Thm.~\ref{thm:equal_to_gpekf} draws equivalence between GP-ZKF's zonotopic estimate and the \textit{unscaled} covariance from GP-EKF.
In practice, we defined GP-EKF, GP-UKF, and GP-PF's set-based estimates as the $95\%$ confidence ellipsoids determined by the \textit{up-scaled} covariance matrices (GP-PF's posteriors were approximated as Gaussians) such that they can be on par with GP-ZKF.
(3) Avg. radius (per dimension): the radius of the box that outer-approximates the set-based estimates, which measures conservativeness.
(4) Avg. computation time per time step.

\begin{figure}[t]
    \centering
    \vspace{4mm}
    \hspace{-4mm}
    \def\svgwidth{0.75\linewidth}
    \input{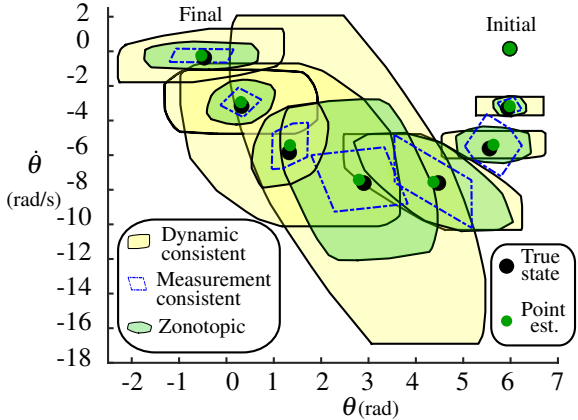}
    \caption{
    Zonotopic estimates along a trajectory produced by GP-ZKF under the \emph{Shift None} condition.
    The zonotopic estimates, $\hat{X}_t$ (green fill), always outer-approximate the intersection between the dynamic-consistent zonotope, $\bar{X}_{t}$ (yellow fill), and the measurement-consistent zonotope, $\bar{X}_{y_t}$ (blue outline).
    (1) Despite the erroneous point estimates, defined as $(\hat{X}_t)_{c}$ (green dots), the zonotopes, $\hat{X}_t$, still contain the true states (black dots).
    (2) Even though the volumes of $\bar{X}_{t}$ may grow rapidly during uncertainty propagation, the volumes of $\hat{X}_t$ shrink when informative measurements are available.
    }
    \label{fig:pendulum_zono_illustration}
    \vspace{-4mm}
\end{figure}

\subsection{Simulated Pendulum Domain}
We simulated a discrete-time 2D inverted pendulum controlled by an infinite-horizon linear quadratic regulator. The state, $\mathbf{x}=[\theta, \dot{\theta}]^T$ (angle and angular velocity), and the set-point correspond to the pendulum standing upright ($0^\circ$). The dynamics, $d(\cdot)$, is the closed-loop system, with noise $\bm{w}$ ($\lambda_w=7.16^\circ$). The observation function, $o(\cdot)$, maps the state to the end effector position and velocity $\in \mathbb{R}^{4}$, with noise $\bm{v}$ ($\lambda_v=8.88^\circ$). The known dynamics, $f(\cdot)$, is given by the linearized and discretized system around the set-point.

To evaluate the methods' consistency under significant epistemic uncertainties, we induced distribution shifts between training and testing. For testing, we ran each method with $T=15$ from four uniformly sampled starting points within the \emph{testing region}, $\theta \in [\pi,2\pi]$, with 10 repetitions per starting point. The training data included four variants:
(1) \emph{Shift Both}: the GPs for both $g(\cdot)$ and $h(\cdot)$ were trained with a \emph{default} dataset with nine rollouts inside the \emph{training region}, $\theta \in [0,\pi]$. The epistemic uncertainties in both the dynamic and observation models were high due to the shift from the training region to the testing region.
(2) \emph{Shift Dynamics}: the GP for $h(\cdot)$ was trained with an additional dataset whose input locations were uniformly sampled in the \emph{testing region}. 
(3) \emph{Shift Observation}: the GP for $g(\cdot)$ was trained with five additional rollouts inside the \emph{testing region}.
(4) \emph{Shift None}: $g(\cdot)$ and $h(\cdot)$ were each trained with their own additional dataset besides the default dataset.
Similar to~\cite{koller2019arXivlearning}, confidence interval scalings $\beta^g$ and $\beta^h$ are conservative. We specified and scaled them when more data was available, to reflect their dependency upon the information capacity, $\gamma$.

We present our results in Tab.~\ref{tab:num_pendulum}.
Due to the nonlinearity of the dynamic and observation functions, the linearization-based method (GP-EKF) performed poorly, with very low inclusions and small radii. GP-UKF (not linearization-based) performed better, with lower RMSEs, higher inclusions and larger radii. GP-PF produced large RMSEs, potentially due to the narrowness of GP's predictive distributions. In contrast, GP-ZKF achieved the highest amount of inclusions for all conditions; this empirically validates GP-ZKF's conservativeness in respecting linearization errors and epistemic and aleatoric uncertainties, which echos GP-ZKF's probabilistic consistency guarantee (Thm.~\ref{thm:consistent}).
Fig.~\ref{fig:pendulum_zono_illustration} illustrates the zonotopic estimates produced by GP-ZKF under the \emph{Shift None} condition. Admittedly, GP-ZKF is more conservative than the others, resulting in larger radii. We argue that GP-ZKF's conservativeness actually scales appropriately with the domain, as we will next demonstrate its low conservativeness in the dressing domain.

\subsection{Robot-Assisted Dressing Domain}\label{ssec:dressing}
A robot arm was controlled to dress a \emph{long-sleeved} jacket onto a human arm (Fig.~\ref{fig:teaser}).
We evaluated all methods offline to estimate the visually occluded positions of human elbow~\cite{zhang2019probabilistic}.

During data collection, the human naturally moved his arm, the configuration of which was tracked by the Xsens motion capture system, which does not suffer from visual occlusion.
The robot was controlled to move from the human hand to the elbow, and then to the shoulder position.
The data is composed of three initial arm conditions: \emph{bend}, \emph{lower}, and \emph{straight}, with $17$, $11$, and $12$ trajectories, respectively.

\newcolumntype{R}{>{\centering\arraybackslash}p{1.775cm}}
\newcolumntype{S}{>{\centering\arraybackslash}p{0.25cm}}
\newcolumntype{O}{>{\centering\arraybackslash}p{0.675cm}}
\newcolumntype{M}{>{\centering\arraybackslash}p{1.1cm}}
\begin{table}[t]
    \vspace{2mm}
    \centering
    \caption{The estimation results in the Robot-Assisted Dressing Domain.}
    \label{tab:num_dressing}
    \begin{tabular}{OMRSRK}
        \toprule
        Arm & Method & RMSE $(x,y,z)$ & Incl. & Radius $(x,y,z)$ & Time \\
        pose &  & (cm,cm,cm) & (\%) & (cm,cm,cm) & (s)\\
        \midrule
        Bend
& GP-EKF & 3.9, 3.3, 3.8 & 83 & 11, 19, 23 & 0.25\\
& GP-UKF & 9.8, 3.0, 6.1 & 76 & 26, 10, 15 & 1.68\\
& GP-PF & 2.4, 2.6, 3.0 & 73 & 6, 6, 7 & 29.93\\
& \textbf{GP-ZKF} & 2.8, 3.4, 4.3 & \textbf{88} & 8, 8, 8 & 0.25\\
        \midrule
        Lower
& GP-EKF & 5.2, 3.0, 6.9 & 94 & 24, 19, 22 & 0.11\\
& GP-UKF & 9.2, 12.4, 13.6 & 80 & 27, 22, 33 & 0.71\\
& GP-PF & 5.0, 2.7, 4.1 & 70 & 11, 8, 11 & 12.74\\
& \textbf{GP-ZKF} & 4.6, 2.7, 4.7 & \textbf{97} & 12, 10, 12 & 0.11\\
        \midrule
        Straight
& GP-EKF & 2.0, 3.0, 3.9 & 88 & 14, 15, 14 & 0.12\\
& GP-UKF & 3.4, 3.8, 3.6 & 64 & 7, 7, 9 & 0.78\\
& GP-PF & 2.5, 3.0, 2.8 & 62 & 6, 6, 6 & 13.95\\
& \textbf{GP-ZKF} & 1.7, 2.8, 4.1 & \textbf{92} & 8, 8, 8 & 0.12\\
        \bottomrule
    \end{tabular}
    \vspace{-4mm}
\end{table}

The state, $\mathbf{x} \in \mathbb{R}^{3}$, is the human elbow position.
The control, $\mathbf{u} \in \mathbb{R}^{9}$,
contains the positions of the human hand and shoulder, and the robot end effector.
The measurement $\mathbf{y} \in \mathbb{R}^{3}$ is the processed force signal, obtained at the robot end effector.
We applied a low-pass filter to the force and converted it to the approximate position of the center of the cuff via a tether-inspired parametric model~\cite{tognon2021physical}.
The known dynamic model is specified as $f(\mathbf{x},\mathbf{u})=\mathbf{x}$;
the variances of the noises, $\mathbf{w}$ and $\mathbf{v}$, are automatically identified within GP regression.

We present the cross-validated results within each initial arm condition in Tab.~\ref{tab:num_dressing}.
Both GP-EKF and GP-UKF achieved high amount of inclusions that were comparable with GP-ZKF;
however, GP-ZKF's radii were much smaller, demonstrating its appropriate conservativeness.
Essentially, GP-ZKF provided consistent set-based estimates without being overly conservative;
\cref{fig:dressing} illustrates the zonotopic estimates it produced.

\section{Conclusion and Future Work}\label{sec:conclusion_future_work}
This work proposes a set-based estimator, GP-ZKF, that is able to produce zonotopic estimates based on learned dynamic and observation models.
Theoretically, GP-ZKF guarantees probabilistic consistency.
In addition, the \emph{stochastic} approach, GP-EKF, can be seen as a special case of our \emph{set-based} method, GP-ZKF---\emph{i.e.}, where linearization errors and aleatoric uncertainties are omitted and epistemic uncertainties are simplified.
Empirically, GP-ZKF outperformed the \emph{stochastic} baselines (GP-EKF, GP-UKF, and GP-PF~\cite{ko2009gp}) in both a simulated pendulum and real-world dressing domain.
Future work will focus on combining our approach with control methods~\cite{LiRSS21,Stou2020Online}, as well as richer and more computationally scalable models for force-based estimation.




\section{APPENDIX}


We combine the confidence intervals for $g(\cdot)$ (Lemma~\ref{lemma:well_calibrated}) with the bounds for $\mu^g$ (Eq.~\ref{eq:predict:learned_function:mean}) and $\sigma^g$ (Eq.~\ref{eq:predict:learned_function:var}) to bound the error $|g_j(\cdot)-\bar{\mu}^g_j(\cdot)|$ for each $j=1,\dots,n_x$ as follows:
\begin{equation}
\resizebox{1\hsize}{!}{$
\begin{aligned}
  &|g_j(\mathbf{x}_{t-1}, \mathbf{u}_{t-1}) - \bar{\mu}^g_j(\mathbf{x}_{t-1}, \mathbf{u}_{t-1})|
  \\
  &\leq
  |\mu^g_j(\mathbf{x}_{t-1}, \mathbf{u}_{t-1}) - \bar{\mu}^g_j(\mathbf{x}_{t-1}, \mathbf{u}_{t-1})|
  +
  |g_j(\mathbf{x}_{t-1}, \mathbf{u}_{t-1}) - \mu^g_j(\mathbf{x}_{t-1}, \mathbf{u}_{t-1})|
  \:\:(b)
  \\
  &\leq \frac{1}{2} L^g_{\nabla \mu} \cdot \|\mathbf{x}_{t-1} - \bar{\mathbf{x}}_{t-1}\|_2^2
  + \beta^g \cdot \sigma^g_{j}(\mathbf{x}_{t-1}, \mathbf{u}_{t-1})
  \:\:(c)
  \\
  &\leq \frac{1}{2} L^g_{\nabla \mu} \cdot \|\mathbf{x}_{t-1} - \bar{\mathbf{x}}_{t-1}\|_2^2
  + \beta^g \cdot \sigma^g_{j}(\bar{\mathbf{x}}_{t-1}, \mathbf{u}_{t-1})
  + \beta^g L^g_{\sigma} \cdot \|\mathbf{x}_{t-1} - \bar{\mathbf{x}}_{t-1}\|_2^{1/2}
  \:\:(d)
  \\
  &\leq \frac{1}{2} L^g_{\nabla \mu} \cdot \|\hat{X}_{t-1}-\bar{\mathbf{x}}_{t-1}\|_2^2
  + \beta^g \cdot \sigma^g_{j}(\bar{\mathbf{x}}_{t-1}, \mathbf{u}_{t-1})
  + \beta^g L^g_{\sigma} \cdot \|\hat{X}_{t-1}-\bar{\mathbf{x}}_{t-1}\|_2^{1/2}
  \:\:(e)
\nonumber
\end{aligned}
$}
\end{equation}
Here, $(b)$ is obtained via the triangle inequality,
and we achieve $(c)$ by using a union bound to combine the mean's error bound (Eq.~\ref{eq:predict:learned_function:mean}) and confidence intervals (Lemma~\ref{lemma:well_calibrated}).
Hence, all inequalities starting at $(c)$ hold with probability at least $(1-\delta^g-\delta^w)$.
$(d)$ is obtained by applying the standard deviation's error bound (Eq.~\ref{eq:predict:learned_function:var});
$(e)$ is due to the norm of the translated zonotope
$\|\hat{X}_{t-1}-\bar{\mathbf{x}}_{t-1}\|_2 \coloneqq \max_{\mathbf{x}_{t-1} \in \hat{X}_{t-1}} \|\mathbf{x}_{t-1} - \bar{\mathbf{x}}_{t-1}\|_2$~\cite{althoff2021cora}.

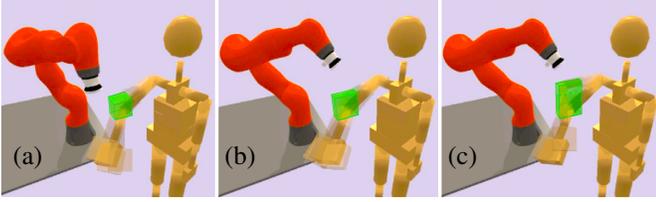
\begin{figure}[t]
 \vspace{2mm}
  \centering
  \def\svgwidth{0.99\linewidth}
\begingroup%
  \makeatletter%
  \providecommand\color[2][]{%
    \errmessage{(Inkscape) Color is used for the text in Inkscape, but the package 'color.sty' is not loaded}%
    \renewcommand\color[2][]{}%
  }%
  \providecommand\transparent[1]{%
    \errmessage{(Inkscape) Transparency is used (non-zero) for the text in Inkscape, but the package 'transparent.sty' is not loaded}%
    \renewcommand\transparent[1]{}%
  }%
  \providecommand\rotatebox[2]{#2}%
  \newcommand*\fsize{\dimexpr\f@size pt\relax}%
  \newcommand*\lineheight[1]{\fontsize{\fsize}{#1\fsize}\selectfont}%
  \ifx\svgwidth\undefined%
    \setlength{\unitlength}{187.68000412bp}%
    \ifx\svgscale\undefined%
      \relax%
    \else%
      \setlength{\unitlength}{\unitlength * \real{\svgscale}}%
    \fi%
  \else%
    \setlength{\unitlength}{\svgwidth}%
  \fi%
  \global\let\svgwidth\undefined%
  \global\let\svgscale\undefined%
  \makeatother%
  \begin{picture}(1,0.29795395)%
    \lineheight{1}%
    \setlength\tabcolsep{0pt}%
    \put(0,0){\includegraphics[width=\unitlength,page=1]{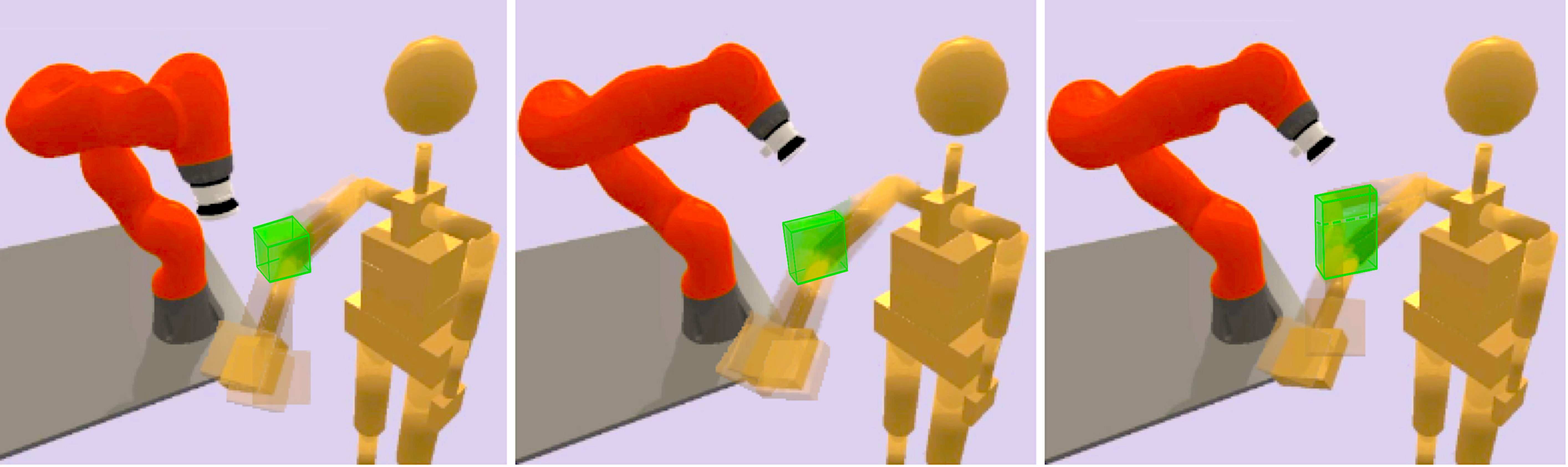}}%
    \put(0.01657429,0.04962152){\color[rgb]{0,0,0}\makebox(0,0)[lt]{\lineheight{0.60000002}\smash{\begin{tabular}[t]{l}(a)\end{tabular}}}}%
    \put(0.33757429,0.04962152){\color[rgb]{0,0,0}\makebox(0,0)[lt]{\lineheight{0.60000002}\smash{\begin{tabular}[t]{l}(b)\end{tabular}}}}%
    \put(0.6757429,0.04962152){\color[rgb]{0,0,0}\makebox(0,0)[lt]{\lineheight{0.60000002}\smash{\begin{tabular}[t]{l}(c)\end{tabular}}}}%
  \end{picture}%
\endgroup%

  \caption{
  Keyframes of the robot-assisted dressing task (cloth is omitted). The simulated robot and human motions were reproduced based on real-world data (Fig.~\ref{fig:teaser}). 
  Each keyframe illustrates a few ground-truth states of the human arm and the zonotopic estimate of the human elbow position (outer-approximated as boxes $\subset \mathbb{R}^{3}$, visualized in green).
  }
  \vspace{-1mm}
  \label{fig:dressing}
  \vspace{-3mm}
\end{figure}

\appendices






%


\bibliographystyle{IEEEtran}
\bibliography{references}


%






\end{document}